\documentclass[12pt,reqno]{amsart}
\usepackage{amsmath}
\usepackage{latexsym}
\usepackage{amsfonts}
\usepackage{amssymb}
\usepackage{color}
\usepackage{bbm,dsfont}
\usepackage{graphicx}

\newtheorem{theorem}{Theorem}

\newtheorem{proposition}{Proposition}
\newtheorem{example}{Example}
\newtheorem{definition}{Definition}

\newcommand{\C}{\mathbb C} 
\newcommand{\half}{\tfrac{1}{2}} 

\newcommand{\hi}{\mathcal{H}} 
\newcommand{\eh}{\mathcal{E(H)}} 
\newcommand{\ip}[2]{\left\langle\,#1\,|\,#2\,\right\rangle} 
\newcommand{\kb}[2]{|#1\rangle\langle#2|} 
\newcommand{\tr}[1]{\textrm{tr}\left[#1\right]} 
\newcommand{\id}{\mathbbm{1}} 
\newcommand{\nul}{0} 

\newcommand{\rank}[1]{\textrm{rank}\left(#1\right)} 

\newcommand{\Ao}{\mathsf{A}}
\newcommand{\Bo}{\mathsf{B}}
\newcommand{\Co}{\mathsf{C}}

\newcommand{\Aspan}{\mathcal{A}}

\newcommand{\R}{\mathbb{R}}

\begin{document}

\title[Relabeling and mixing]{Quantum measurements on finite dimensional systems: relabeling and mixing}

\author[Haapasalo]{Erkka Haapasalo}
\email{ethaap@utu.fi}

\author[Heinosaari]{Teiko Heinosaari}
\email{teiko.heinosaari@utu.fi}

\author[Pellonp\"a\"a]{Juha-Pekka Pellonp\"a\"a}
\email{juhpello@utu.fi}


\begin{abstract}
Concentrating on finite dimensional systems, we show that one can limit to extremal rank-1 POVMs if two simple procedures of mixing and relabeling are permitted.
We demonstrate that any finite outcome POVM can be obtained from extremal rank-1 POVMs with these two procedures.
In particular, extremal POVMs with higher rank are just relabelings of extremal rank-1 POVMs and their structure is therefore clarified.
\end{abstract}


\maketitle

\section{Introduction}\label{sec:intro}

Quantum information theory has raised new questions related to the conventional quantum formalism. 
One particular idea is the minimal resource perspective.
Generally, the underlying question is the following.
If we want to realize all devices of the given type and our methods are specified, what are those devices that allow one to construct all desired devices? 

For instance, suppose we want to be able to prepare all pure states of a bipartite system. 
If we are able to prepare all maximally entangled pure states, then by means of local operations and classical communication (LOCC) we can prepare all pure states. 

These type of questions and their variants have been investigated extensively in the case of quantum states. 
Similar minimal resource investigations of quantum measurements are less common. 
From the mathematical point of view, quantum measurements are described by positive operator valued measures (POVMs) and questions therefore reduce to the problems on the mathematical structure of POVMs.

In this work we concentrate on one type of minimal resource problem in the case of finite outcome POVMs on a finite dimensional system.
If we can implement two POVMs, then we can implement their mixture by randomly alternating the measurements. 
An extremal POVM corresponds to a measurement that cannot be obtained as a mixture.
For a qubit system all extremal POVMs are of rank-1 \cite{DaLoPe05} and one can thus restrict to rank-1 POVMs when searching for an optimal measurement for some task (see e.g. \cite{VeBe10}).
Generally, however, there are extremal POVMs which are not rank-1 and their relevance is somewhat unclear, even if they have been mathematically characterized \cite{DaLoPe05},  \cite{Parthasarathy99}, \cite{Pellonpaa11}.

In this work we show that one can limit to extremal rank-1 POVMs, assuming that two simple procedures are permitted.
It turns out that it is enough to allow another simple method than mixing; this is called \emph{relabeling}. 
We show that any POVM can be obtained from extremal rank-1 POVMs with these two procedures of mixing and relabeling.
This result summarizes the structure of POVMs in a very convenient and accesible way.

We remark that some related results have been derived in \cite{MaMu90a}.
In particular, it was shown that any finite outcome POVM on a finite dimensional system can be obtained from a rank-1 POVM by a stochastic matrix.
This procedure can be interpreted as a postprocessing of obtained measurement outcomes \cite{BuDaKePeWe05}.
In this context, relabeling is just a deterministic processing method, i.e., the related stochastic matrix contains only entries $0$ and $1$.
Since a double stochastic matrix is a convex mixture of permutation matrices \cite{bir}, one can easily see a connection between these two approaches.

It was argued in \cite{MaMu90a} that maximality in the postprocessing relation is physically more relevant than extremality.
This may be the case, but we believe that a solid understanding of the extremal POVMs is important since in many problems one has to maximize a convex figure of merit.
Our aim in this work is to clarify the relations between all POVMs, extremal POVMs and extremal rank-1 POVMs.

\section{Preliminaries}

Quantum measurements are generally described by positive operator valued measures (POVMs).
In this work we concentrate on measurements with finite number of outcomes. 
In Sec. \ref{sec:general} we make some remarks on POVMs with infinite number of outcomes.

Let $\hi$ be a finite $d$-dimensional Hilbert space.
We denote by $\eh$ the set of all selfadjoint operators satisfying the operator inequalities $\nul\leq E\leq\id$; these are called \emph{effects}.
A POVM is an assignment of an effect for each measurement outcome \cite{PSAQT82}.
The particular labeling of measurement outcomes is irrelevant for our investigation. 
We will therefore assume that the measurement outcomes are labeled by the integers $\{1,\ldots, N\}$.
Thus, a POVM $\Ao$ with $N$ outcomes is a function $j\mapsto \Ao(j)$ from the outcome set $\Omega_N\equiv \{1,\ldots, N\}$ to the set of effects  $\eh$ and it is required to satisfy the normalization condition $\sum_{j=1}^N \Ao(j) = \id$.

We remark that it is possible that $\Ao(j)=\nul$ for some $j$. 
This simply means that the outcome $j$ is never registered.
Two POVMs that differ only by their number of zero operators are considered equivalent. 

A special class of POVMs are those that consists of projections. 
We say that a POVM $\Ao$ is a \emph{projection valued measure} (PVM) if each $\Ao(j)$ is a projection, that is, $\Ao(j)^2 = \Ao(j)$ for every $j$.

Another special class of POVMs are rank-1 POVMs.
A POVM $\Ao$ is called \emph{rank-1} if all its (nonzero) effects $\Ao(j)$ are rank-1 operators \cite{rank1}.
A PVM consisting of one-dimensional projections is rank-1, but there are also other rank-1 POVMs (we give examples in Subsec. \ref{sec:mix}).
We make a simple observation on the number of outcomes in a rank-1 POVM.

\begin{proposition}\label{prop:Nd}
Suppose $\Ao$ is a rank-1 POVM and has $N$ nonzero outcomes. Then $N\geq d$.
If $N=d$, then $\Ao$ is a PVM.
\end{proposition}

\begin{proof}
Since 
\begin{eqnarray*}
N &=&\rank{\Ao(1)} + \cdots +\rank{\Ao(N)} \\
&\geq & \rank{\Ao(1) +\cdots +\Ao(N)} \\
&=& \rank{\id}=d \, , 
\end{eqnarray*}
we conclude that $N\geq d$.

Suppose $N=d$. Then
\begin{equation}\label{eq:N=d}
d = \tr{\id}= \sum_{j=1}^d \tr{\Ao(j)} \, .
\end{equation}
Since $\Ao(j)$ is rank-1, we have $0< \tr{\Ao(j)} \leq 1$.
Thus, \eqref{eq:N=d} implies that $\tr{\Ao(j)}=1$ for every $j$.
It follows that each $\Ao(j)$ is a one-dimensional projection \cite{rank1}.
\end{proof}

\section{Relabeling and mixing of POVMs}\label{sec:mix}

There are two basic procedures to obtain a new POVM out of the given ones, relabeling and mixing.
We will first consider some of their properties and then show that, if taken together, we can construct any POVM from extremal rank-1 POVMs by using these two methods.

\subsection{Relabeling}

By \emph{relabeling} we mean a procedure where the labels of the measurement outcomes are shuffled, possibly giving same label to several different outcomes.
If the resulting effects are kept fixed or only the number of zero effects is modified, we consider the new POVM to be equivalent to the initial POVM.
In particular, this kind of transformation is reversible since we can relabel the new POVM again to get the initial POVM back.

In irreversible transformations same label is given to several different nonzero outcomes; this means that the corresponding effects are added together.
Namely, suppose a POVM $\Ao$ has $N$ outcomes.
We define a new POVM $\Ao'$ with $M<N$ outcomes by identifying two or several outcomes.
For instance, if the outcomes $1$ and $2$ are identified and the rest are just renamed, then the resulting POVM $\Ao'$ has $N-1$ outcomes and it is given by
\begin{eqnarray*}
\Ao'(1) &=& \Ao(1) + \Ao(2) \, , \\
\Ao'(j) &=& \Ao(j+1) \, , \quad j=2,\ldots N-1 \, .
\end{eqnarray*}
This is depicted in Fig. \ref{fig:relabeling}.

\begin{figure}
\begin{center}
\includegraphics[width=8.0cm]{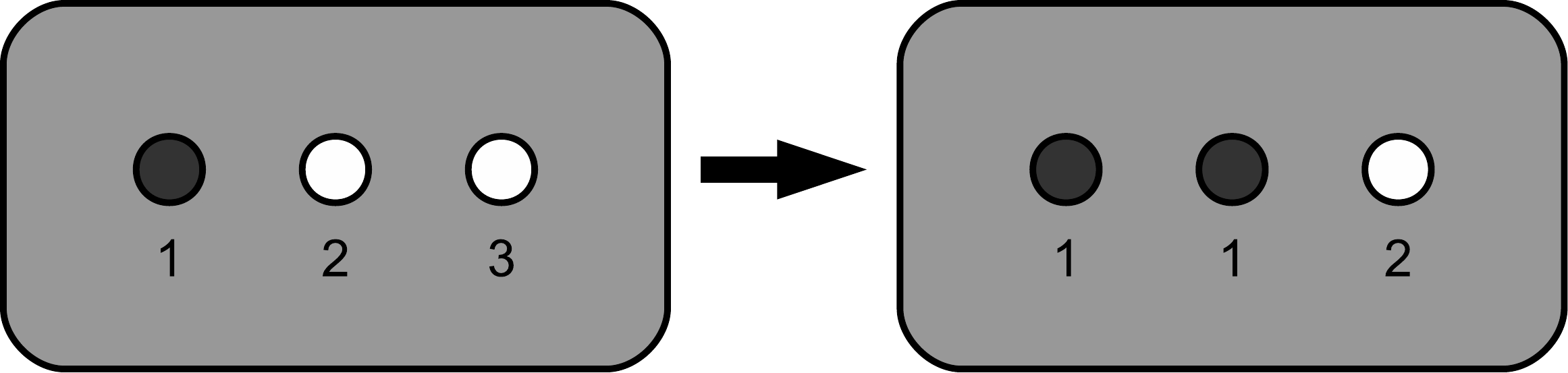}
\end{center}
\caption{\label{fig:relabeling} 
A POVM can be thought of as a box with LEDs indicating the obtained measurement outcome.
The measurement outcome '1' is recorded if the first LED flashes.
After the relabeling, the measurement outcome '1' is recorded if either the first or the second LED flashes.}
\end{figure}

A general definition can be formulated as follows.

\begin{definition}(\emph{Relabeling})
A POVM $\Ao$ can be obtained from another POVM $\Bo$ by relabeling of outcomes if there is a function $f:\Omega_N\to\Omega_M$ such that
\begin{equation*}
\Ao(j) = \sum_{k \in f^{-1}(j)} \Bo(k) \, .
\end{equation*}
(Here $f^{-1}(j)$ is the preimage of $j$, i.e., $f^{-1}(j)=\{ k\in\Omega_N : f(k)=j \}$.)
\end{definition}

It is now an easy consequence of the spectral theorem that all POVMs are relabelings of rank-1 POVMs.

\begin{proposition}\label{prop:relabel}
A POVM $\Ao$ with $N$ outcomes is a relabeling of a rank-1 POVM $\Bo$ with $M\leq N\cdot d$ outcomes.
\end{proposition}

\begin{proof}
For each $j$, we write the effect $\Ao(j)$ in its spectral decomposition form
\begin{equation*}
\Ao(j) = \alpha_{j1} P_{j1} + \cdots \alpha_{jd} P_{jd} \, .
\end{equation*}
Here $\alpha_{j1},\ldots,\alpha_{jd}$ are the eigenvalues (counting multiplicities) of $\Ao(j)$ and $P_{j1},\ldots,P_{jd}$ are orthogonal one-dimensional projections.
We define a POVM $\Bo$ on the product set $\Omega_N\times\Omega_d$ by
\begin{equation*}
\Bo(j,k) = \alpha_{jk} P_{jk} \, .
\end{equation*}
Since 
\begin{equation*}
\Ao(j)=\sum_{k=1}^d \Bo(j,k) = \sum_{(j,k) \in f^{-1}(j)} \Bo(j,k)
\end{equation*}
for the projection function $f:\Omega_N\times\Omega_d\to\Omega_N$, $f(j,k)=j$,
we conclude that $\Ao$ can be obtained from $\Bo$ by relabeling.
Every nonzero effect of $\Bo$ is of rank-1.
Therefore, by dropping zero effects we obtain a rank-1 POVM with $M\leq Nd$ outcomes.
\end{proof}

Prop. \ref{prop:relabel} alone is not useful from the minimal resource aspect since the number of outcomes of the required rank-1 POVMs depends on the implemented POVM.

\subsection{Mixing}\label{sec:mix}

Another basic procedure is to alternate two measurements in a random fashion.
Suppose $\Bo$ and $\Co$ are two POVMs.
We can start from the situation that $\Bo$ has $N$ outcomes and $\Co$ has $M\leq N$ outcomes.
We then add enough zero operators to $\Co$ so that it has also $N$ outcomes. 
A mixture of $\Bo$ and $\Co$ is defined as
\begin{equation}\label{eq:mix}
\Ao(j) = t \Bo(j) + (1-t) \Co(j) \, , \quad j=1,\ldots,N \, ,
\end{equation}
where $0<t<1$ is the mixing weight.

\begin{definition}(\emph{Mixture})
A POVM $\Ao$ is a \emph{mixture} of two POVMs $\Bo$ and $\Co$ if \eqref{eq:mix} holds for some $0<t<1$.
An \emph{extremal} POVM cannot be given as a mixture of two different POVMs.
\end{definition}

In the following we list two basic facts on extremal POVMs. 
These are easy consequences of the known characterizations \cite{DaLoPe05},  \cite{Parthasarathy99}, \cite{Pellonpaa11}.
For reader's convenience, we provide direct proofs that do not require any characterization results.

\begin{proposition}\label{prop:basic}
Let $\Ao$ be a POVM consisting of $N$ nonzero effects.
\begin{itemize}
\item[(a)] If $\Ao$ is extremal, then the effects $\Ao(1),\ldots,\Ao(N)$ are linearly independent and $N\leq d^2$.
\item[(b)] If $\Ao$ is rank-1 and the effects $\Ao(1),\ldots,\Ao(N)$ are linearly independent, then $\Ao$ is extremal.
\end{itemize}
\end{proposition}

\begin{proof}
\begin{itemize}
\item[(a)] Suppose that the effects $\Ao(1),\ldots,\Ao(N)$ are linearly dependent. 
We will show that $\Ao$ is not extremal.
The linear dependence means that
\begin{equation}
\lambda_1 \Ao(1) + \cdots + \lambda_N \Ao(N) = 0
\end{equation}
for some coefficients $\lambda_j \in \R$ which are not all zeroes. 
Notice that since $\Ao(j)$ are positive operators, there has to be both positive and negative $\lambda_j$'s.
We can assume that $\lambda_1$ is the greatest (hence positive) coefficient and $\lambda_N$ is the smallest (hence negative) coefficient.
We define two POVMs $\Ao'$ and $\Ao''$ by
\begin{equation}\label{eq:mix1}
\left\{
\begin{array}{ll}
\Ao'(1) &= \nul  \\ 
\Ao'(j) &= (1-\lambda_j/\lambda_1) \Ao(j) \quad j=2,\ldots,N
\end{array} \right.
\end{equation}
and
\begin{equation}\label{eq:mix2}
\left\{
\begin{array}{ll}
\Ao''(j) &= (1-\lambda_j/\lambda_N) \Ao(j) \quad j=1,\ldots,N-1\\
\Ao''(N) &= \nul \, .
\end{array} \right.
\end{equation}

Setting $t=\lambda_1/(\lambda_1 - \lambda_N)$ we obtain
\begin{equation}\label{eq:mix3}
\Ao(j) = t \Ao'(j) + (1-t) \Ao''(j)
\end{equation}
for every $j=1,\ldots,N$.
Thus, $\Ao$ is a mixture of $\Ao'$ and $\Ao''$.
Since $\Ao'(N)\neq\nul=\Ao''(N)$, these are two different POVMs and $\Ao$ is not extremal.

The dimension of the real vector space of all selfadjoint operators is $d^2$, hence $N\leq d^2$.

\item[(b)] Suppose that $\Ao=\lambda \Bo + (1-\lambda) \Co$ for some number $0<\lambda<1$ and some observables $\Bo$ and $\Co$.
This implies that $\lambda\Bo(j)\leq\Ao(j)$ and $(1-\lambda)\Co(j)\leq\Ao(j)$ for every $j=1,\ldots,N$.
As each $\Ao(j)$ is a rank-1 operator, it follows that $\Bo(j)=b_j\Ao(j)$ and $\Co(j)=c_j \Ao(j)$ for some non-negative numbers $b_j,c_j\in\R$. 
Since $\Bo$ and $\Co$ satisfy the normalization $\sum_j \Bo(j)=\sum_j \Co(j)=\id$, we obtain 
\begin{equation*}
\sum_{j=1}^N b_j \Ao(j)=\sum_{j=1}^N c_j\Ao(j)=\id \, .
\end{equation*}
Therefore
\begin{equation*}
\qquad \nul = \id-\id = \sum_{j=1}^N (1-b_j) \Ao(j) = \sum_{j=1}^N (1-c_j) \Ao(j)\, . 
\end{equation*}
The operators $\Ao(1),\ldots,\Ao(N)$ are linearly independent, thus $b_j=c_j=1$ for every $j$.
But this means that $\Ao=\Bo=\Co$, hence $\Ao$ is not a mixture of two different POVMs and it is thus extremal.
\end{itemize}
\end{proof}

If $\dim\hi=2$, then it can be shown that every extremal POVM is rank-1 \cite{DaLoPe05}.
In all higher dimensional Hilbert spaces there are PVMs that contain projections with rank greater than 1. 
The projections are the extremal effects \cite{QTOS76}, and it thus follows that all PVMs are extremal \cite{PSAQT82}. 
At this point we are interested on extremal rank-1 POVMs, but we will comment more on different types of extremal POVMs in Sec. \ref{sec:example}.

Let us first make some observations on extremal rank-1 POVMs.
Combining Prop. \ref{prop:Nd} and Prop. \ref{prop:basic} we conclude that {\it the number of nonzero outcomes $N$ of an extremal rank-1 POVM is between $d$ and $d^2$.}
It is easy to give examples of extremal rank-1 POVMs having the minimal number of outcomes $N=d$.

\begin{example}\label{ex:pvm}(\emph{Extremal rank-1 POVM with $d$ outcomes})
Fix an orthonormal basis $\{\varphi_j\}_{j=1}^d$ for $\hi$.
We define a $d$-outcome POVM $\Ao$ by $\Ao(j)=\kb{\varphi_j}{\varphi_j}$, and
each effect $\Ao(j)$ is hence a one-dimensional projection.
Since $\Ao$ is a PVM, it is extremal. 
Alternatively, the extremality of $\Ao$ can be concluded using Prop. \ref{prop:basic};
if $\sum_{j=1}^d c_j \kb{\varphi_j}{\varphi_j} = 0$, then $0=\sum_{j=1}^d c_j \kb{\varphi_j}{\varphi_j}\varphi_k=c_k\varphi_k$ for every $k$. 
Hence $c_1=\cdots=c_d=0$. 
\end{example}

Using the previous example as an auxiliary result, we can prove the following.

\begin{proposition}\label{prop:exist}
For every $N$ satisfying $d\leq N \leq d^2$, there exists an extremal rank-1 POVM $\Ao$ with $N$ outcomes. 
\end{proposition}

\begin{proof}
Suppose the claim is true for a number $N$ satisfying $d\leq N < d^2$.
We show that the claim is then true also for $N+1$.

Let  $\Ao$ be an extremal rank-1 POVM $\Ao$ with $N<d^2$ outcomes. 
Since $N<d^2$, the real linear span $\Aspan$ of the effects $\Ao(1),\ldots,\Ao(N)$ is a proper subset of the real linear space of selfadjoint operators.
Therefore, there exists a nonzero selfadjoint operator $S$ which is not in $\Aspan$.
The operator $S$ has a spectral decomposition $S=\sum_j \alpha_j P_j$, and at least one of the spectral projections $P_j$ is not in $\Aspan$ since otherwise $S$ would be in $\Aspan$.
We can thus choose a one-dimensional projection $P$ which is not in $\Aspan$ and we denote $T=\id + P$.
Then $T$ is positive and invertible, hence it has invertible square root $T^{\half}$ .
We define
\begin{eqnarray*}
\Ao'(j) &=& T^{-\half}\Ao(j)T^{-\half} \qquad j=1,\ldots,N\\
\Ao'(N+1)&=& T^{-\half}PT^{-\half}
\end{eqnarray*}
The resulting POVM $\Ao'$ is rank-1 and the effects $\Ao'(1),\ldots,\Ao'(N+1)$ are linearly independent.

The claim now follows from this consideration combined with Example \ref{ex:pvm}.
\end{proof}

We have seen that by starting from a $d$-outcome extremal rank-1 POVM, we can construct extremal rank-1 POVMs with any allowed number of outcomes.
In the following example we demonstrate the method used in the proof of Prop. \ref{prop:exist}.

\begin{example}\label{ex:qubit-1}
We take $\hi=\C^2$ and construct a three-outcome extremal rank-1 POVM from a two-outcome extremal rank-1 POVM.
By Prop. \ref{prop:Nd} a rank-1 POVM $\Ao$ with two outcomes is necessarily a PVM. 
Let us consider $\Ao$ with $\Ao(1)=\half(\id+\sigma_x)$, $\Ao(2)=\half(\id-\sigma_x)$.
Using the same notations as  in the proof of Prop. \ref{prop:exist}, we can choose $P=\half(\id+\sigma_z)$ and this gives $T^{-\half}=c_+ \id + c_-\sigma_z$ with $c_\pm=(1\pm\sqrt{2})/2\sqrt{2}$.
The resulting extremal POVM $\Ao'$ is therefore
\begin{eqnarray*}
\Ao'(1) &=& T^{-\half}\Ao(1)T^{-\half} = \frac{3}{8} \bigl( \id + \frac{4}{3\sqrt{2}}\sigma_x - \frac{1}{3} \sigma_z \bigr) \\
\Ao'(2) &=& T^{-\half}\Ao(2)T^{-\half} = \frac{3}{8} \bigl( \id - \frac{4}{3\sqrt{2}}\sigma_x - \frac{1}{3} \sigma_z \bigr)\\
\Ao'(3)&=& T^{-\half}PT^{-\half} = \frac{1}{4} \bigl( \id+\sigma_z \bigr) \, .
\end{eqnarray*}
This is extremal by construction, and one can also directly verify that the effects are rank-1 and linearly independent.
\end{example}

\subsection{Implementation of every POVM}

As we have recalled in Subsec. \ref{sec:mix}, the number of outcomes of extremal rank-1 POVMs ranges from $d$ to $d^2$ and the extremality criterion is simply the linear independence of its effects. 

In general, the extremality criterion is the following \cite{DaLoPe05}, \cite{Parthasarathy99}, \cite{Pellonpaa11}:
Let $\Ao$ be a POVM with nonzero effects $\Ao(j)$, $j=1,\ldots,N$. 
We can write 
\begin{equation}\label{eq:spectral}
\Ao(j)=\sum_{k=1}^{n(j)}\kb{\psi_k(j)}{\psi_k(j)} \, ,
\end{equation}
where the vectors $\psi_k(j)$, $k=1,\ldots,n(j)$, are nonzero and orthogonal.
Then $\Ao$ is extremal if and only if the operators $\kb{\psi_k(j)}{\psi_\ell(j)}$  for  $j=1,\ldots,N$ and $k,\,\ell=1,\ldots,n(j)$ are linearly independent.

We conclude that the set of all extremal rank-1 POVMs is much more tractable and easier to handle than the set of all extremal POVMs. 
It would be therefore desirable to concentrate only on extremal rank-1 POVMs, but then one has to justify this limitation in some way.
The following two results show that we can restrict to extremal rank-1 POVMs if we are allowed to perform their relabeling and mixing.

\begin{theorem}\label{th:every}
Every POVM can be obtained from extremal rank-1 observables by mixing and relabeling. 
\end{theorem}

\begin{proof}
Let $\Ao$ be a $N$-outcome POVM.
By Prop. \ref{prop:relabel} there is a rank-1 POVM $\Bo$ such that $\Ao$ is a relabeling of $\Bo$.
If the effects of $\Bo$ are linearly independent, then $\Bo$ is extremal and the claim is true.
Hence, we concentrate on the case that the effects of $\Bo$ are linearly dependent.

If a rank-1 observable $\Bo$ has $M$ nonzero outcomes and the corresponding effects are linearly dependent, then we can write $\Bo$ as a mixture of two observables $\Bo'$ and $\Bo''$, both consisting of (at most) $M-1$ nonzero rank-1 effects.
The proof is similar as in Prop. \ref{prop:basic}; the linear dependence means that
\begin{equation}
\lambda_1 \Bo(1) + \cdots + \lambda_M \Bo(M) = 0
\end{equation}
for some coefficients $\lambda_j \in \R$ which are not all zeroes. 
Choosing similarly as in formulas \eqref{eq:mix1}-\eqref{eq:mix3}, we can write $\Bo$ as a mixture of two different POVMs $\Bo'$ and $\Bo''$, both having at most $M-1$ nonzero effects.
All their nonzero effects are rank-1 since all nonzero effects of $\Bo$ are rank-1.

If $\Bo'$ or $\Bo''$ consist of linearly dependent effects, we continue this procedure. 
Since the number of nonzero effects is always decreased by one, the process has to terminate. 
We end up with extremal rank-1 POVMs.  
\end{proof}

Let us remark that the proof of Theorem \ref{th:every} shows that to implement a given POVM $\Ao$ which is not extremal rank-1, we first mix and then relabel.
Namely, in the first step we mix two or more extremal rank-1 POVMs to obtain a POVM with more outcomes than $\Ao$. 
In the second step we relabel the measurement outcomes in a suitable way and in this way obtain $\Ao$.

As one would expect, the first step of mixing is not needed if the desired POVM $\Ao$ is extremal.
(This is perhaps not transparent from Theorem \ref{th:every} since mixing is used only as an intermediate step.)
We have the following result.

\begin{theorem}\label{th:every-ext}
Every extremal POVM is either rank-1 or a relabeling of an extremal rank-1 POVM.
\end{theorem}

\begin{proof}
Assume that $\Ao$ is an extremal POVM with nonzero effects $\Ao(j)$, $j=1,\ldots,N$. 
We write each $\Ao(j)$ in the form \eqref{eq:spectral}.
The extremality of $\Ao$ is equivalent to the fact that the operators $\kb{\psi_k(j)}{\psi_l(j)}$, $j=1,\ldots,N$, $k,\,l=1,\ldots,n(j)$, are linearly independent.
We define a rank-1 POVM $\Bo$ by
$\Bo(j,k)=\kb{\psi_k(j)}{\psi_k(j)}$ where $j=1,\ldots,N$, $k=1,\ldots,n(j)$.
Obviously $\Ao$ is a relabeling of $\Bo$.
Since the effects $\Bo(j,k)$ are linearly independent, $\Bo$ is extremal.
\end{proof}

Theorem \ref{th:every-ext} indicates that we can produce other extremal POVMs from an extremal rank-1 POVM by relabeling its outcomes.
However, we emphasize that not every relabeling leads to an extremal POVM.
For instance, if we start from the three-outcome POVM $\Ao'$ written in Example \ref{ex:qubit-1} and relabel it to obtain a two-outcome POVM, then the resulting POVM is not extremal.
On the other hand, if we start from any PVM, then its arbitrary relabeling is still a PVM and therefore extremal.

\section{Four types of extremal POVMs}\label{sec:example}

The extremality condition, rephrased in \eqref{eq:spectral}, can be used to decide whether a given POVM is extremal or not.
However, it leaves open whether there is a simpler characterization of extremal elements.

First of all, we have two basic types of extremal POVMs:
\begin{itemize}
\item[(a)] a rank-1 POVM with linearly independent elements;
\item[(b)] a PVM (of any rank).
\end{itemize}
A third type is a hybrid of the previous ones:
\begin{itemize}
\item[(c)] a POVM such that each effect $\Ao(j)$ is either a projection or a rank-1 operator, and the the operators $\Ao(1),\ldots,\Ao(N)$ are linearly independent.
\end{itemize}

It is easy to verify that all POVMs of (c)-type are extremal. 
(A slight modification of the proof of Prop. \ref{prop:basic}b works).
Clearly, (a) and (b) are just special instances of (c).
We also recall that for $\dim\hi=2$ every extremal POVM is rank-1 \cite{DaLoPe05}, i.e., in this case there are only (a)-type extremal POVMs.

The obvious question is: are all extremal POVMs of the (c)-type?
In the following we demonstrate that this is not the case.
Hence, there exists
\begin{itemize}
\item[(d)] an extremal POVM which is not (c)-type.
\end{itemize}
Our example uses a Hilbert space $\hi$ with $\dim\hi=4$, but with a small modification one can generate a similar example also in all higher dimensions. 

\begin{example}\label{ex:d}
Denote the third roots of 1 by $\omega_3^j$, $j=1,\,2,\,3$, i.e. $\omega_3=e^{i2\pi/3}$. Consider the case $\hi=\C^4$ and pick an orthonormal basis $\{|1\rangle,\,|2\rangle,\,|3\rangle,\,|4\rangle\}$ for $\C^4$. Define operators $\Ao(j)\in\mathcal L(\C^4)$, $j=1,\,2,\,3$,
\begin{eqnarray*}
\Ao(j)&=&\frac13\Big(\id+\omega_3^j\big(|1\rangle\langle3|+|2\rangle\langle4|\big)+\overline{\omega_3^j}\big(|3\rangle\langle1|+|4\rangle\langle2|\big)\Big)\\
&=&\frac13\big(|\psi_1(j)\rangle\langle\psi_1(j)|+|\psi_2(j)\rangle\langle\psi_2(j)|\big),
\end{eqnarray*}
where $\psi_1(j)=|1\rangle+\overline{\omega_3^j}|3\rangle$ and $\psi_2(j)=|2\rangle+\overline{\omega_3^j}|4\rangle$, $j=1,\,2,\,3$. 
It is straightforward to verify that the assignment $\{1,\,2,\,3\}\ni j\mapsto\Ao(j)$ defines a POVM.

The extremality characterization \eqref{eq:spectral} states that $\Ao$ is extremal iff the operators $E_j^{kl}=|\psi_k(j)\rangle\langle\psi_l(j)|$,
\begin{eqnarray*}
E_j^{11}&=&|1\rangle\langle1|+|3\rangle\langle3|+\omega_3^j|1\rangle\langle3|+\overline{\omega_3^j}|3\rangle\langle1|,\\
E_j^{22}&=&|2\rangle\langle2|+|4\rangle\langle4|+\omega_3^j|2\rangle\langle4|+\overline{\omega_3^j}|4\rangle\langle2|,\\
E_j^{12}&=&|1\rangle\langle2|+|3\rangle\langle4|+\omega_3^j|1\rangle\langle4|+\overline{\omega_3^j}|3\rangle\langle2|,\\
E_j^{21}&=&(E_j^{12})^*,
\end{eqnarray*}
$j=1,\,2,\,3$, are linearly independent. 
This is easily seen to be the case using the fact that the equations 
\begin{equation*}
\sum_{j=1}^3\lambda_j=0 \, , \quad \sum_{j=1}^3\omega_3^j\lambda_j=0 \, , \quad \sum_{j=1}^3\overline{\omega_3^j}\lambda_j=0
\end{equation*}
are satisfied by complex numbers $\lambda_1,\,\lambda_2$ and $\lambda_3$ iff $\lambda_1=\lambda_2=\lambda_3=0$.

It is obvious that the operators $\Ao(j)$ are all of rank 2. 
Since $\Ao(j)^2=\frac23\Ao(j)$, these operators are not projections but scalar multiples of projections.
Hence $\Ao$ is not of (c)-type.
\end{example}

We recall that a POVM $\Ao$ with two outcomes is extremal iff $\Ao$ is a PVM \cite{QTOS76}.
Therefore, an extremal POVM of (d)-type has at least three outcomes, and Example \ref{ex:d} is in this sense minimal.
We leave it as an open question whether (d)-type extremal POVMs exists in dimension three.

\section{POVMs on infinite outcome set}\label{sec:general}

In this section we make some remarks that are related to POVMs on infinite outcome set.
These observations are more technical than the previous results.

Generally, the set of measurement outcomes need not be finite. 
It can be, for instance, the set of all real numbers $\R$.
In this case, one has to specify not only the outcome set $\Omega$ but also the $\sigma$-algebra $\Sigma$ of subsets of $\Omega$.
The pair $(\Omega,\Sigma)$ is called an \emph{outcome space}.
In the general formulation a POVM $\Ao$ is a mapping from a $\sigma$-algebra $\Sigma$ to the set of effects $\eh$, and it is required to satisfy the normalization
$\Ao(\Omega)=\id$ and the $\sigma$-additivity $\Ao(\cup X_j) = \sum_j \Ao(X_j)$ for every sequence of disjoint sets $X_j\in\Sigma$.

All POVMs with a fixed outcome space $(\Omega,\Sigma)$ form a convex set.
A characterization of the extremal elements has been derived in a recent work \cite{Pellonpaa11}.
Some of the features of extremal POVMs with finite outcome set have direct generalizations in this general context.
For instance,
if $\Ao$ is an extremal POVM and $X_1,\ldots,X_N\in\Sigma$ are disjoint sets such that $\Ao(X_j)\neq\nul$, then the effects $\Ao(X_1),\ldots,\Ao(X_N)$ are linearly independent (see p.\ 6 in \cite{Pellonpaa11}).
Note that this is a generalization of the well-known result written in Prop. \ref{prop:basic} and it holds also in the case of an infinite dimensional Hilbert space.
From this fact follows that, if $\Ao$ is extremal then there are at most $d^2$ disjoint sets $X_j$ such that $\Ao(X_j)\ne 0$. Hence, an extremal $\Ao$ is concentrated on the set $\cup_{j=1}^N X_j$, $N\le d^2$, but it does not necessarily follow that $\Ao$ is discrete (a POVM $\Ao$ is said to be discrete if there exists a finite or countably infinite set $X\in\Sigma$ such that $\Ao(X)=\id$).
We demonstrate this in Example  \ref{ex:cofinite} below.

\begin{example}\label{ex:cofinite}
Let $\Omega$ be a uncountable set (for instance the real line $\R$ or the interval $[0,1]$) and $\Sigma$ the cofinite $\sigma$-algebra on $\Omega$.
This means that a set $X\subset\Omega$ belongs to $\Sigma$ if and only if $X$ is either countable or the complement of $X$ is countable.

Let $\hi$ be a finite dimensional Hilbert space.
We fix an element $x\in\Omega$ and a projection $P\in\eh$, $0\neq P \neq \id$.
We then define a POVM $\Ao$ on $\Sigma$ by the conditions $\Ao(\{x\})=P$, $\Ao(\{y\})=0$ whenever $y\neq x$ and $\Ao(\Omega)=\id$.
This implies that $\Ao(\Omega\smallsetminus\{x\})=\id-P$.
Since $\Ao$ is PVM, it is extremal. 
But $\Ao$ is not discrete since $\Ao(X)$ is either $P$ or $0$ for any countable set $X$ (depending on whether $x$ is in $X$ or not).
\end{example}

The cofinite $\sigma$-algebra used in Example \ref{ex:cofinite} is the Borel $\sigma$-algebra of $\Omega$ equipped with the cofinite topology. 
In this topology a set $X\subseteq\Omega$ is open if the complement $\Omega\setminus X$ is finite.
For example, when $\Omega=\R$ the cofinite topology is coarser than the usual topology of $\R$.
The cofinite topology is not Hausdorff if $\Omega$ is infinite, and one can claim that this is therefore not a physically motivated example. 
In any case, it seems to be worthwhile to pinpoint the assumptions that lead to the conclusion that extremal POVMs on a finite dimensional system are discrete.

Suppose that $\Omega$ is a Hausdorff topological space and $\Sigma$ is the Borel  $\sigma$-algebra of $\Omega$.
We can then speak about the support of $\Ao$; a point $x\in\Omega$ belongs to the support of $\Ao$ if $\Ao(U_x)\neq\nul$ for every open set $U_x\subset\Omega$ containing $x$.
Since $\Omega$ is Hausdorff, any finite number of points have disjoint neighborhoods.
Our earlier discussion thus yields a conclusion that \emph{if $\Ao$ is extremal, then its support contains at most $d^2$ points.}

The support of $\Ao$ is a closed set, hence its complement set $V$ is open. 
We want to have $\Ao(V)=0$ in order to conclude that $\Ao$ is discrete.
This conclusion can be obtained by several different additional assumptions, either on $\Omega$ or on $\Ao$. 

For instance, suppose that $\Omega$ is second countable, implying that every open cover of a subset $X\subseteq\Omega$ has a countable subcover \cite{GT75}.
In this case the complement set $V$ is a countable union of sets of measure zero, therefore $\Ao(V)=0$.

Another independent assumption giving the desired conclusion is that $\Ao$ is \emph{inner regular}. 
This means that for every $\psi\in\hi$ and open set $U\subseteq\Omega$, we have
\begin{equation*}
\ip{\psi}{\Ao(U)\psi} = \sup \{ \ip{\psi}{\Ao(K)\psi} : K\subseteq U,\ \textrm{$K$ compact}  \} \, .
\end{equation*}
This, again, leads to $\Ao(V)=0$, and we thus reach the result first proved in \cite{ChDaSc07} (using the assumptions that $\Omega$ is a locally compact Hausdorff space and $\Ao$ is regular); \emph{if $\Ao$ is extremal, then it is discrete}.
More precisely, we can write
\begin{equation*}
\Ao(X)=\sum_{j=1}^N \chi_X(x_j) E_j
\end{equation*}
for some finite number $N\leq d^2$ of elements $x_1,\ldots,x_N\in\Omega$ and effects $E_1,\ldots,E_N$.

\section{Conclusion}\label{sec:conclusion}

We have demonstrated that any finite outcome POVM on a finite dimensional system can be obtained from extremal rank-1 POVMs using mixing and relabeling, while any extremal one can be obtained from extremal rank-1 POVMs using relabeling only.
Therefore, even if relabeling is a very simple procedure, it allows to summarize the structure of POVMs in a convenient and accessible form.

It is clear that the mathematical structure of POVMs on an infinite dimensional system is considerably more complicated. 
Further understanding will be helpful in understanding which POVMs are relevant when searching for an optimal measurement for some task.
This problem will be studied elsewhere.

\section*{Acknowledgements}

The authors wish to thank Kari Ylinen for useful discussions. 
This work has been supported by Academy of Finland and The Emil Aaltonen Foundation.


\begin{thebibliography}{10}
\bibitem{bir}
This result is known as the Birkhoff - von Neumann theorem.
\bibitem{rank1}
A rank-1 operator $E$ can be written in the form $E=eP$ for some number $0<e\leq 1$ and a one-dimensional projection $P=\kb\psi\psi$ where $\|\psi\|=1$. Notice that $\tr{E}=e$.
\end{thebibliography}
\end{document}